\documentclass[american,aps,pra,reprint,groupedaddress,twocolumn,showpacs]{revtex4-1} 
\usepackage[unicode=true,pdfusetitle, bookmarks=true,bookmarksnumbered=false,bookmarksopen=false, breaklinks=false,pdfborder={0 0 0},backref=false,colorlinks=false] {hyperref}
\hypersetup{ colorlinks,linkcolor=myurlcolor,citecolor=myurlcolor,urlcolor=myurlcolor}
\usepackage{braket,colortbl,amsthm,amsmath,amssymb,txfonts}
\theoremstyle{plain}
\newtheorem{thm}{Theorem}
\newtheorem{exmp}{Example}
\definecolor{myurlcolor}{rgb}{0,0,0.7}
\usepackage{graphics,graphicx}
\usepackage{color}
\usepackage{colortbl}
\usepackage{subfigure}
\usepackage{dirtytalk}
\usepackage{float}
\usepackage{hyperref}
\usepackage[font=small,labelfont=bf,
   justification= centerlast,
   format=hang]{caption}

\usepackage{ragged2e}
\usepackage[utf8x]{inputenc}
\usepackage{color}
\usepackage{braket}
\usepackage{physics}
\usepackage{tikz}
\usetikzlibrary{shapes.geometric, arrows}
\usepackage{latexsym}
\usepackage{bm}
\usepackage{graphics,epstopdf}
\usepackage{color}
\usepackage[usenames,dvipsnames,svgnames]{pstricks}
\usepackage{epsfig}
\usepackage{pst-grad} % For gradients
\usepackage{pst-plot} % For axes
\begin{document}
	\title{Characterizing the generalized complementarity polytope with extractable information from MUBs}
	\author{Gautam Sharma}
	\email{gautam.oct@gmail.com}
	\affiliation{Optics and Quantum Information Group, Institute of Mathematical Sciences,\\HBNI, CIT Campus, Taramani, Chennai 600113, India}

\begin{abstract}
	Complementarity polytope is a geometric structure that exists in $N^2-1$ dimensional space for an $N$ dimensional Hilbert space. The existence of $N+1$ mutually unbiased bases(MUBs) is possible, if such a polytope can be shown to be a subset of density matrices, which is a very difficult task. With the hope of simplifying this task, we have shown in this work that, the complementarity polytope can be characterized by the total extractable information from $N+1$ MUBs. We also demonstrate that $t\leq N+1$ number of MUBs also form a polytope that exists in $t(N-1)$ dimensional space, which we refer to as ``generalized complementarity polytope". The generalized complementarity polytope can also be characterized by total extractable information from $t$ MUBs.
\end{abstract}

\maketitle
\section{Introduction}
Quantum systems can have properties which are complementary to each other, i.e., if we know completely about one property, then we have zero knowledge about it's complementary property. Mathematically, this characteristic is represented via the mutual unbiasedness of the bases of complementary observables. Two orthonormal bases $A=\{\ket{a_1},...,\ket{a_n}\}$ and $B=\{\ket{b_1},...,\ket{b_n}\}$ in an $N$ dimensional hilbert state $\mathcal{H}_N$ are said to be mutually unbiased if 
\begin{align*}
\forall k,l	\hspace*{1cm}|\braket{a_k}{b_l}|^2=\frac{1}{N}.
\end{align*}

It is known that a maximum of $N+1$ MUBs can exist in a Hilbert space of dimension $N$ \cite{bengtsson2007three}. However, only when $N$ is a prime or power of prime number $n=p^k$,  $N+1$ such bases have been found \cite{wootters1989optimal,ivonovic1981geometrical}. In fact, it is a well studied topic and the properties of MUBs for prime powered dimensions are well understood \cite{brierley2009all,PhysRevA.70.012302,PhysRevA.72.062310,combescure2009block,wiesniak2011entanglement,mcnulty2016mutually}. For cases, when $N$ is not a prime powered number the maximum number of MUBs is not known \cite{wocjan2005new,durt2010mutually,alber2018mutually,horodecki2020five}. In fact, for the simplest composite dimension of $N=6$, solutions containing a maximum of only three bases have been found \cite{bengtsson2007mutually,butterley2007numerical,PhysRevA.79.052316,grassl2004sic,jaming2009generalized,jaming2010problem,PhysRevA.83.062303,goyeneche2013mutually,batle2013new,chen2018mutually,mcnulty2012all,mcnulty2012impossibility}.

Among many efforts to find MUBs in dimensions which are not power of prime, a geometric object called the complementary polytope was discovered in \cite{bengtsson2005mutually}. A complementarity polytope exists in an $N^2-1$ dimensional space for all dimensions regardless of whether $N$ is power of prime or not. However for $N+1$ MUBs to exist it is necessary that the complementarity can be accommodated inside the polytope of density matrices. However, using this approach also, it was not possible to pinpoint on any peculiarity which will solve this problem.

In this work, we have introduced a new measure  of information and use it to calculate extractable information from a set of $N+1$ MUBs. We show that this total information characterizes the complementarity polytope. We further show the existence of generalized complementarity polytopes for $t\leq N+1$ number of MUBs and characterize it using total extractable information from $t$ MUBs. Through the generalized polytopes it could be easier to show the existence or non-existence of $t \leq N+1$ MUBs.

\section{Absolute Measure of Information and Uncertainty}
In quantum theory, the measurement of an observable on a quantum system is described in terms of probabilities of the possible outcomes. Any such probability distribution $\{p_n\}_{n=1}^N$ of the $N$ possible outcomes is least informative whenever it is a maximally uncertain distribution, i.e., $\forall n, p_n=\frac{1}{N}$. Also, $\{p_n\}_{n=1}^N$ is most informative when it is a maximally certain probability distribution. Therefore, we can define the information content of a probability distribution as $\mathbb{I}_p=1-\mathbb{U}_p$, where $\mathbb{U}_p$ can be any proper measure of uncertainty. The uncertainty measures can be constructed using Schur-concave functions of probability distributions \cite{PhysRevLett.111.230401}. In our case we use the following measure of uncertainty for the probabability distribution $\{p_n\}_{n=1}^N$ 
\begin{align*}
\mathbb{U}=1-\frac{N}{2(N-1)}\sum_{n=1}^{N}\bigg|p_n-\frac{1}{N}\bigg|,
\end{align*}
where $\frac{N}{2(N-1)}$ is the normalization so that $0\leq\mathbb{U}\leq1$. Using this measure of uncertainty we define it's absolute information content as

\begin{align}\label{inf}
\mathbb{I}=\frac{N}{2(N-1)}\sum_{n=1}^{N}\bigg|p_n-\frac{1}{N}\bigg|.
\end{align}
It is straightforward to note that $0\leq \mathbb{I}\leq 1$. We highlight here that, this measure of information has a very similar structure to how the Brukner and Zeilinger(BZ) invariant information  \cite{PhysRevLett.83.3354} was constructed. The only difference being that we are taking the absolute value of $p_n-\frac{1}{N}$, instead of squaring it. One can also define information using any other measure of uncertainty also \cite{hardy1952inequalities}, but as we shall show that only this form is useful for our purpose. 

\section{Complementarity polytope from Information theoretic perspective}
That, a complete set of MUBs in $N$ dimensions defines a complementarity polytope in $N^2-1$ dimensions, was originally shown in \cite{bengtsson2005mutually}. It was also concluded that the existence of $N+1$ set of MUBs is dependent on the fact whether the complementarity polytope can be made to form a subset of density matrices. Here, we show that the total extractable information from MUBs characterizes the complementarity polytope. To demonstrare this, first we briefly present the structure of states in $N^2-1$ dimensions \cite{bengtsson2005mutually}.

We look for how the pure MUB states sit in the space of density matrices in $N^2-1$ dimensions. To do that, we first define the squared distance between two matrices $\rho_1$ and $\rho_2$ as 
\begin{align}\label{distance}
	D^2(\rho_1,\rho_2)=\frac{1}{2}\Tr(\rho_1-\rho_2)^2.
\end{align}
which ensures that the states occupy an $N^2-1$ dimensional space. In the $N^2-1$ dimensional space the matrices are represented with vectors, such that the origin lies at the matrix $\rho_*=\frac{\mathbb{I}}{N}$ with the scalar product given as
\begin{align}\label{scalarproduct}
	(\rho_1,\rho_2)=\frac{1}{4}[D^2(\rho_1+\rho_2,\rho_*)-D^2(\rho_1-\rho_2,\rho_*)]=\frac{1}{2}\bigg[\Tr\rho_1\rho_2-\frac{1}{N}\bigg].
\end{align}

Let's assume that $\mathcal{O}_m=\{\Pi_{nm}\}_{n=1}^N$ and $\mathcal{O}_{m'}=\{\Pi_{nm'}\}_{n=1}^N$ with $m,m'\in \{1,...,N+1\}$, are two orthonormal bases which are mutually unbiased with each other. Note that in the $N^2-1$ dimensional space, the projectors $\Pi_{nm}$ need not always represent quantum states, they can be matrices lying on the boundary of the sphere in $N^2-1$ dimensions. The projectors of these MUBs satisfy the following relations 
\begin{align}
&\Tr (\Pi_{nm})^2=1,  \label{projector}\\
&\Tr \Pi_{nm}{\Pi}_{n'm}=0,  \hspace{1cm}   n\neq{n'}, \label{orthogonality} \\
&\Tr \Pi_{nm}{\Pi}_{nm'}=\frac{1}{d},  \hspace{1cm}   m\neq{m'}. \label{mutualunbiasedness}
\end{align}

From Eq.(\ref{projector}), it follows that all the pure states are located on the surface of a sphere of radius $\sqrt{\frac{N-1}{2N}}$ centered at $\rho_*$. This sphere forms the outsphere of the convex set of density matrices. Further Eq.(\ref{distance}) along with Eq.(\ref{orthogonality}) implies that two different projectors from an orthonormal basis are at unit distance from each other. Thus, all the states belonging to an orthonormal basis $\mathcal{O}_m=\{\Pi_{nm}\}_{n=1}^N$, form a regular simplex in an $N-1$ dimensional space. Finally, on supplying Eq.(\ref{mutualunbiasedness}) in Eq.(\ref{scalarproduct}), tells us that the vectors representing mutually unbiased vectors are orthogonal with each other, implying that two different MUBs lie in orthogonal subspaces.

Therefore, the $N^2-1$ dimensional space can be divided into $N+1$ number of $N-1$ dimensional spaces which implies that we can have a maximum of $N+1$ MUBs for $N$ dimensional Hilbert space. In this way we get $N(N+1)$ number of points forming a convex polytope in a $N^2-1$ dimensional space. 

If we restrict ourselves to prime or prime power dimensions, i.e., $N=p^k$, where $p$ is a prime number and $k$ is an integer, then it is known that there exist $N+1$ MUBs. But the same can't be said about Hilber spaces with composite dimensions. Despite all this, a complementarity polytope exists irrespective of whether $N$ is prime or non-prime.

To calculate the total extractalbe information from MUBs for a given state $\rho$, we note that each MUB ${\mathcal{O}_m}$  is associated with the probability distribution $\{p_{nm}\}_{n=1}^{N}$ with information content given by (from Eq.(\ref{inf}))
\begin{align*}
	\mathbb{I}_m=\frac{N}{2(N-1)}\sum_{n=1}^{N}\bigg|p_{nm}-\frac{1}{N}\bigg|.
\end{align*}
Further, as $\mathcal{O}_m$ and $\mathcal{O}_{m'}$ are mutually unbiased, they capture completely independent information about $\rho$, i.e., $\mathbb{I}_m$ and $\mathbb{I}_{m'}$ are independent. Then, all the MUBs together can capture total information of $\rho$, which is given by
\begin{align}\label{inftotal}
\mathbb{I}_{total}&=\sum_{j=1}^{N+1}\mathbb{I}_m,\nonumber \\
&=\frac{N}{2(N-1)}\sum_{m=1}^{N+1}\sum_{n=1}^{N}\bigg|p_{nm}-\frac{1}{N}\bigg|.
\end{align}

Now, we present the main result of our work.

\begin{thm}
	{The states on the surface of complementarity polytope are characterized by $\mathbb{I}_{total}=1$. Where as the states inside and outside the complementarity polytope are identified via $\mathbb{I}_{total}<1$ and $\mathbb{I}_{total}>1$ respectively.}\label{mainresult} \newline
\end{thm} 
\begin{proof}
	A state on the surface of the complementarity polytope can be written as a convex hull of projectors $\Pi_{nm}$ of the MUBs, such that at least one projector from each MUB is not included. A state on this surface is represented by $\mathcal{\rho}_{surface}=\sum_{m=1}^{N+1}\sum_{n=1}^{N}a_{nm}\Pi_{nm}$, so that $\sum_{m=1}^{N+1}\sum_{n=1}^{N}a_{nm}=1$. In this summation, atleast one $a_{nm}$ is zero for each MUB $\mathcal{O}_m$. For this state the information content from probability distribution of an MUB $\mathcal{O}_{m=s}$ is given by
	\begin{align}\label{infofrommub}
	\mathbb{I}_s=\sum_{n=1}^{N}\frac{N}{2(N-1)}\bigg(\big|\frac{a_{ns}(N-1)}{N}\big|+(N-1)\big|\frac{a_{ns}}{N}\big|\bigg)=\sum_{n=1}^{N}a_{ns}.
	\end{align}
	and therefore we have the total information from MUBs on the surface of the complementarity polytope as 
	\begin{align}\label{totalinfo}
	\mathbb{I}_{total}=\sum_{s=1}^{N+1}\mathbb{I}_s=\sum_{s=1}^{N+1}\sum_{n=1}^{d}a_{ns}=1.
	\end{align}
	
	Thus, on the surface of the complementarity polytope the total information content from the MUBs is unity everywhere. Any point above the surface outside the polytope can also be represented as a sum of $\rho_{out}=\sum_{m=1}^{N+1}\sum_{n=1}^{N}a_{nm}\Pi_{nm}$, with $\sum_{m=1}^{N+1}\sum_{n=1}^{N}a_{nm}>1$. Hence, the total extractable information from MUBs, for any state outside the complementarity polytope is greater than unity, i.e. $\mathbb{I}_{total}>1$. Using a similar argument, for a point below this surface and inside the polytope will have $\mathbb{I}_{total}<1$. This concludes the proof.
\end{proof}

This feature of constant value of $\mathbb{I}_{total}$ on the surface of complementarity polytope, has a stark similarity with the BZ invariant information, which has constant value on the surface of the $N^2-1$ dimensional sphere, which can be seen from the form of BZ information and Eq.(\ref{projector})
\begin{align*}
	\mathbb{I}_{BZ}=\sum_{i=1}^{N+1}\bigg(p_i-\frac{1}{N}\bigg)^2=2\Tr\rho^2-\frac{1}{N}.
\end{align*}
In fact, as we take different values of $\alpha$ in $\bigg|p_i-\frac{1}{N}\bigg|^{\alpha}$, we can get a different surface with constant total information from MUBs.  It should also be noted that the $\mathbb{I}_{total}$ from Eq.(\ref{inftotal}) is not invariant with respect to the choice of complementarity set of bases, unlike the BZ information. Thus, for each choice of set of complementarity bases, $\mathbb{I}_{total}$ identifies a different complementarity polytope.

Coming to the information theoretic characterization of the complementary polytope via total extractable information from MUBs. The total information from MUBs is constant on the surface of complementarity polytope is always true irrespective of whether the dimension $N$ is prime powered or not. It remains unclear how this information can help to verify the existence or non-existence of $N+1$ MUBs.

\section{Generalized complementarity polytopes}

Now, we will show that for Hilbert space of dimension $N$, even less than $N+1$ number of MUBs form a polytope which we refer to as the ``Generalized complementarity polytope"(GCP). We keep the setting from the previous section the same, i.e., the notion of distance(Eq.\ref{distance}) and scalar product(Eq.\ref{scalarproduct}) are unchanged. The only different thing we do is to consider only $t$ number of MUBs with $t\leq N+1$, so that the $t$ MUBs $\mathcal{O}_m$ with $m\in \{1,...,t\}$, occupy a $t(N-1)$ dimensional space. An analysis similar to previous section follows. 

Each orthonormal basis $\mathcal{O}_m=\{\Pi_{nm}\}_{n=1}^N$ forms a regular simplex consuming $N-1$ number of dimensions. Also, two different MUBs lie in orthogonal spaces, and therefore the $t$ MUBs lie in the $t(N-1)$ dimensional space. In this way we have $tN$ number points from $t$ MUBs forming a convex polytope. It should be noted that all the generalized polytopes with  $t\leq N+1$, always have its vertices on the surface of the $N^2-1$ dimensional sphere.

Not only that, the generalized polytopes can also be characterized with an information measure, however unlike Theorem.\ref{mainresult}, this time the information measure is the total information extractable with only $t$ MUBs, which is defined as
\begin{align}\label{infgen}
\mathbb{I}_{t-total}&=\sum_{m=1}^{t}\mathbb{I}_m,\nonumber \\
&=\frac{N}{2(N-1)}\sum_{m=1}^{t\leq N+1}\sum_{n=1}^{N}\bigg|p_{nm}-\frac{1}{N}\bigg|.
\end{align}

\begin{thm}\label{mainresult2}
	{The states on the surface of generalized complementarity polytopes are characterized by $\mathbb{I}_{t-total}=1$. While the states inside and outside the complementarity polytope are identified via $\mathbb{I}_{t-total}<1$ and $\mathbb{I}_{t-total}>1$ respectively.}
\end{thm} 
\begin{proof}
	The proof is  exactly same as for Theorem.\ref{mainresult}, but we still write it here for completeness. A state on the surface of a GCP can be written as $\mathcal{\rho}_{surface}=\sum_{m=1}^{t}\sum_{n=1}^{N}a_{nm}\Pi_{nm}$, so that $\sum_{m=1}^{t}\sum_{n=1}^{N}a_{nm}=1$, with atleast one $a_{nm}$ is zero for each MUB $\mathcal{O}_m$. The information extractable from each of the MUB is given by Eq.(\ref{infofrommub}). On adding, we can get the extractable information from $t$ MUBs as 
   \begin{align}\label{ttotalinfo}
	\mathbb{I}_{t-total}=\sum_{s=1}^{t}\mathbb{I}_s=\sum_{s=1}^{t}\sum_{n=1}^{N}a_{ns}=1.
	\end{align}
	Following the proof of Theorem.(\ref{mainresult}), we can also show that for the states inside and outside the GCP, $\mathbb{I}_{t-total}<1$ and $\mathbb{I}_{t-total}>1$ respectively.
\end{proof}

The utility of generalized polytopes can be understood as following. Given an $N^2-1$ dimensional Euclidean space, a complementarity polytope always exists and if it's possible to rotate it so that it becomes a subset of set of density matrices implies the existence of $N+1$ MUBs in $N$ dimensions. In the same way a generalized complementarity always exists in a $t(N-1)$(and hence in $N^2-1$) dimensions, and if it can be arranged such that it becomes a subset of density matrices, implies the existence of $t$ number of MUBs in $N$ dimensions. If such an arrangement is not possible, it implies the non-existence of  $t$ or greater than $t$ number of MUBs. Therefore, we need to handle far lesser number of parameters to prove the non-existence of $N+1$ number of MUBs. 

\begin{exmp}
For the simplest case two-dimensional Hilbert space we can see in Fig.\ref{octasquare} the two possible generalized polytopes. For three dimensional Hilbert space, the smallest GCP occupies a 4 dimensional space, with 6 vertices. In general, the smallest GCP for an $N$ dimensional Hilbert space, occupies $2(N-1)$ dimensions.
\begin{figure}[H]
	\subfigure[]{
		\includegraphics[width=0.24\textwidth, keepaspectratio]{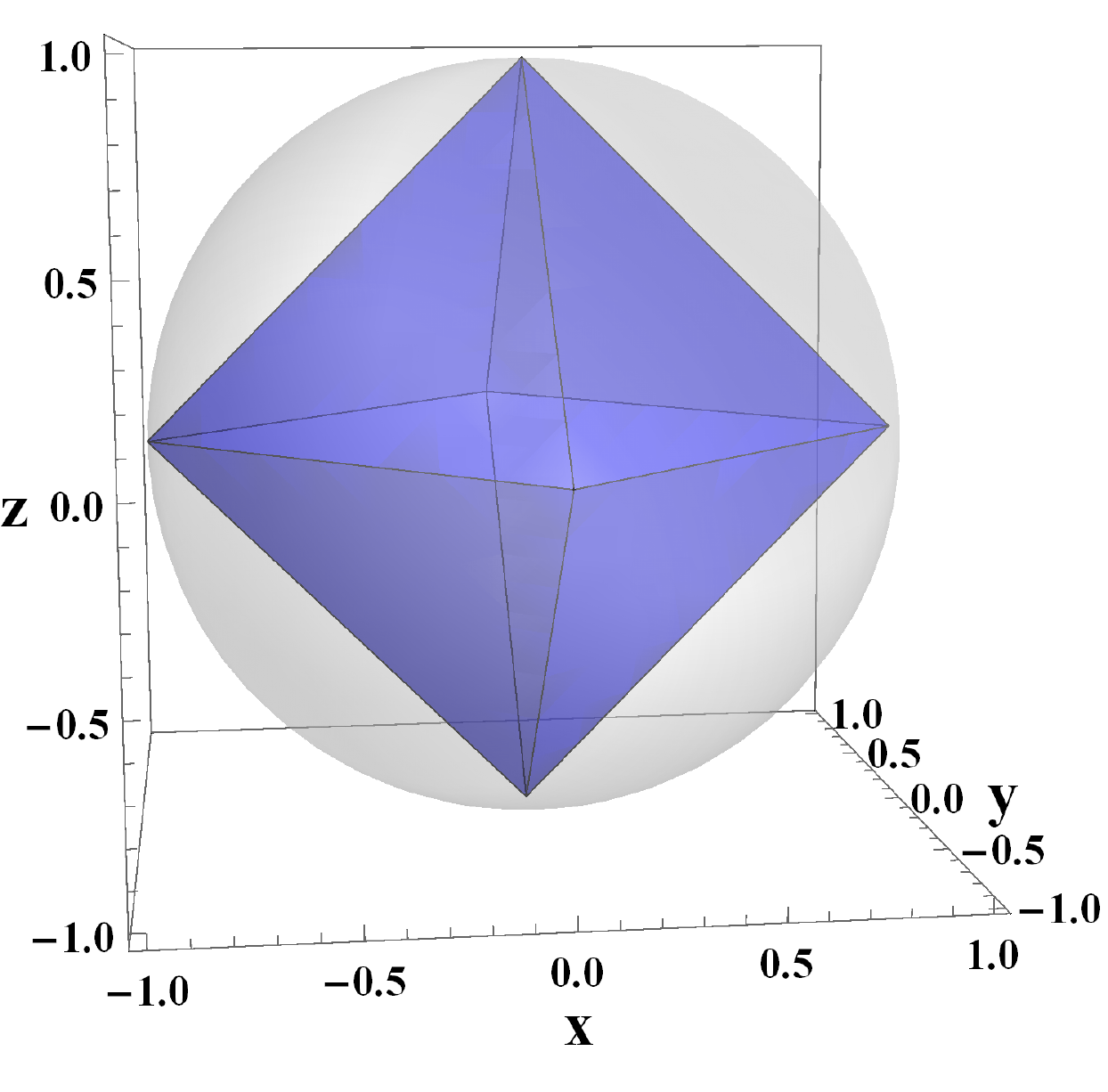}
		
	}
	\subfigure[]{
		\includegraphics[width=0.22\textwidth, keepaspectratio]{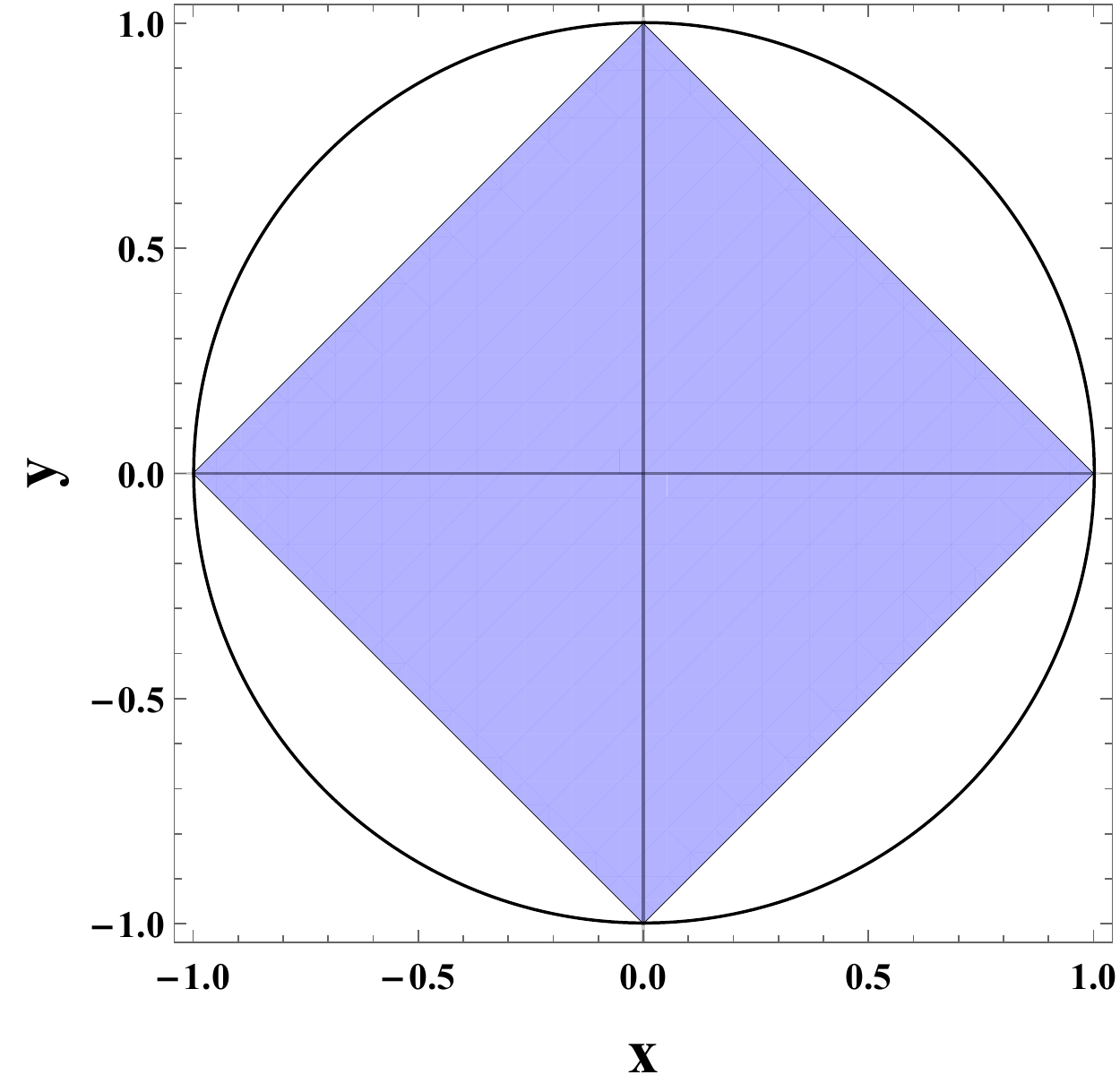}       
		
	}
	\caption{(Color online)Two possible complementarity polytopes for qubit states. We get an octahedron for 3 MUBs and a square for 2 MUBs.}
	\label{octasquare}
\end{figure}
\end{exmp}

\section{Discussion and conclusion}
To summarize the work done here. We introduce the absolute information measure of a probability distribution and calculate the total information extractable from MUBs using this measure. Thereafter, we successfully characterized the complementarity polytope using the total extractable information from MUBs. Further, we demonstrated that $t\leq N+1$ number of MUBs also form a polytope, which we refer to as the generalized complementarity polytope. Moreover, the GCP can also be characterized by the extractable information from $t$ MUBs. With the GCP, it should be a lot easier to check whether it can be made to lie within the convex set of density matrices and whether $t$ MUBs exist for a given Hilbert space.

So far, we couldn't recognize how the total extractable information could help us to study the existence or non-existence of $N+1$ number of MUBs in $N$ dimensional Hilbert space. We hope it leads to better understanding of the geometry of quantum states and MUBs. To look for the existence of $t\leq N+1$ should be a much easier task than to look for the existence of $N+1$ MUBs. GCPs can be made use of in this direction. One way to approach it could be to look for the ratios of the volume of GCP formed with $t$ MUBs and the volume of all density matrices in the $t(N-1)$ dimensional subspace. Finding the volume of a GCP is a straightforward task, but the volume of all density matrices in $t(N-1)$ dimensions is a complicated task, so we have left it for future work.  

\bibliographystyle{apsrev4-1}
\bibliography{compolyref}

\begin{thebibliography}{29}%
\makeatletter
\providecommand \@ifxundefined [1]{%
 \@ifx{#1\undefined}
}%
\providecommand \@ifnum [1]{%
 \ifnum #1\expandafter \@firstoftwo
 \else \expandafter \@secondoftwo
 \fi
}%
\providecommand \@ifx [1]{%
 \ifx #1\expandafter \@firstoftwo
 \else \expandafter \@secondoftwo
 \fi
}%
\providecommand \natexlab [1]{#1}%
\providecommand \enquote  [1]{``#1''}%
\providecommand \bibnamefont  [1]{#1}%
\providecommand \bibfnamefont [1]{#1}%
\providecommand \citenamefont [1]{#1}%
\providecommand \href@noop [0]{\@secondoftwo}%
\providecommand \href [0]{\begingroup \@sanitize@url \@href}%
\providecommand \@href[1]{\@@startlink{#1}\@@href}%
\providecommand \@@href[1]{\endgroup#1\@@endlink}%
\providecommand \@sanitize@url [0]{\catcode `\\12\catcode `\$12\catcode
  `\&12\catcode `\#12\catcode `\^12\catcode `\_12\catcode `\%12\relax}%
\providecommand \@@startlink[1]{}%
\providecommand \@@endlink[0]{}%
\providecommand \url  [0]{\begingroup\@sanitize@url \@url }%
\providecommand \@url [1]{\endgroup\@href {#1}{\urlprefix }}%
\providecommand \urlprefix  [0]{URL }%
\providecommand \Eprint [0]{\href }%
\providecommand \doibase [0]{http://dx.doi.org/}%
\providecommand \selectlanguage [0]{\@gobble}%
\providecommand \bibinfo  [0]{\@secondoftwo}%
\providecommand \bibfield  [0]{\@secondoftwo}%
\providecommand \translation [1]{[#1]}%
\providecommand \BibitemOpen [0]{}%
\providecommand \bibitemStop [0]{}%
\providecommand \bibitemNoStop [0]{.\EOS\space}%
\providecommand \EOS [0]{\spacefactor3000\relax}%
\providecommand \BibitemShut  [1]{\csname bibitem#1\endcsname}%
\let\auto@bib@innerbib\@empty
%</preamble>
\bibitem [{\citenamefont {Bengtsson}(2007)}]{bengtsson2007three}%
  \BibitemOpen
  \bibfield  {author} {\bibinfo {author} {\bibfnamefont {I.}~\bibnamefont
  {Bengtsson}},\ }in\ \href
  {https://aip.scitation.org/doi/abs/10.1063/1.2713445?journalCode=apc} {\emph
  {\bibinfo {booktitle} {AIP Conference Proceedings}}},\ Vol.\ \bibinfo
  {volume} {889}\ (\bibinfo {organization} {American Institute of Physics},\
  \bibinfo {year} {2007})\ pp.\ \bibinfo {pages} {40--51}\BibitemShut {NoStop}%
\bibitem [{\citenamefont {Wootters}\ and\ \citenamefont
  {Fields}(1989)}]{wootters1989optimal}%
  \BibitemOpen
  \bibfield  {author} {\bibinfo {author} {\bibfnamefont {W.~K.}\ \bibnamefont
  {Wootters}}\ and\ \bibinfo {author} {\bibfnamefont {B.~D.}\ \bibnamefont
  {Fields}},\ }\href
  {https://www.sciencedirect.com/science/article/pii/0003491689903229?via%3Dihub}
  {\bibfield  {journal} {\bibinfo  {journal} {Annals of Physics}\ }\textbf
  {\bibinfo {volume} {191}},\ \bibinfo {pages} {363} (\bibinfo {year}
  {1989})}\BibitemShut {NoStop}%
\bibitem [{\citenamefont {Ivonovic}(1981)}]{ivonovic1981geometrical}%
  \BibitemOpen
  \bibfield  {author} {\bibinfo {author} {\bibfnamefont {I.}~\bibnamefont
  {Ivonovic}},\ }\href
  {https://iopscience.iop.org/article/10.1088/0305-4470/14/12/019} {\bibfield
  {journal} {\bibinfo  {journal} {Journal of Physics A: Mathematical and
  General}\ }\textbf {\bibinfo {volume} {14}},\ \bibinfo {pages} {3241}
  (\bibinfo {year} {1981})}\BibitemShut {NoStop}%
\bibitem [{\citenamefont {Brierley}\ \emph {et~al.}(2010)\citenamefont
  {Brierley}, \citenamefont {Weigert},\ and\ \citenamefont
  {Bengtsson}}]{brierley2009all}%
  \BibitemOpen
  \bibfield  {author} {\bibinfo {author} {\bibfnamefont {S.}~\bibnamefont
  {Brierley}}, \bibinfo {author} {\bibfnamefont {S.}~\bibnamefont {Weigert}}, \
  and\ \bibinfo {author} {\bibfnamefont {I.}~\bibnamefont {Bengtsson}},\ }\href
  {https://dl.acm.org/doi/abs/10.5555/2011464.2011470} {\bibfield  {journal}
  {\bibinfo  {journal} {Quantum Info. Comput.}\ }\textbf {\bibinfo {volume}
  {10}},\ \bibinfo {pages} {803–820} (\bibinfo {year} {2010})}\BibitemShut
  {NoStop}%
\bibitem [{\citenamefont {Lawrence}(2004)}]{PhysRevA.70.012302}%
  \BibitemOpen
  \bibfield  {author} {\bibinfo {author} {\bibfnamefont {J.}~\bibnamefont
  {Lawrence}},\ }\href {\doibase 10.1103/PhysRevA.70.012302} {\bibfield
  {journal} {\bibinfo  {journal} {Phys. Rev. A}\ }\textbf {\bibinfo {volume}
  {70}},\ \bibinfo {pages} {012302} (\bibinfo {year} {2004})}\BibitemShut
  {NoStop}%
\bibitem [{\citenamefont {Romero}\ \emph {et~al.}(2005)\citenamefont {Romero},
  \citenamefont {Bj\"ork}, \citenamefont {Klimov},\ and\ \citenamefont
  {S\'anchez-Soto}}]{PhysRevA.72.062310}%
  \BibitemOpen
  \bibfield  {author} {\bibinfo {author} {\bibfnamefont {J.~L.}\ \bibnamefont
  {Romero}}, \bibinfo {author} {\bibfnamefont {G.}~\bibnamefont {Bj\"ork}},
  \bibinfo {author} {\bibfnamefont {A.~B.}\ \bibnamefont {Klimov}}, \ and\
  \bibinfo {author} {\bibfnamefont {L.~L.}\ \bibnamefont {S\'anchez-Soto}},\
  }\href {\doibase 10.1103/PhysRevA.72.062310} {\bibfield  {journal} {\bibinfo
  {journal} {Phys. Rev. A}\ }\textbf {\bibinfo {volume} {72}},\ \bibinfo
  {pages} {062310} (\bibinfo {year} {2005})}\BibitemShut {NoStop}%
\bibitem [{\citenamefont {Combescure}(2009)}]{combescure2009block}%
  \BibitemOpen
  \bibfield  {author} {\bibinfo {author} {\bibfnamefont {M.}~\bibnamefont
  {Combescure}},\ }\href {https://aip.scitation.org/doi/10.1063/1.3078420}
  {\bibfield  {journal} {\bibinfo  {journal} {Journal of mathematical physics}\
  }\textbf {\bibinfo {volume} {50}},\ \bibinfo {pages} {032104} (\bibinfo
  {year} {2009})}\BibitemShut {NoStop}%
\bibitem [{\citenamefont {Wie{\'s}niak}\ \emph {et~al.}(2011)\citenamefont
  {Wie{\'s}niak}, \citenamefont {Paterek},\ and\ \citenamefont
  {Zeilinger}}]{wiesniak2011entanglement}%
  \BibitemOpen
  \bibfield  {author} {\bibinfo {author} {\bibfnamefont {M.}~\bibnamefont
  {Wie{\'s}niak}}, \bibinfo {author} {\bibfnamefont {T.}~\bibnamefont
  {Paterek}}, \ and\ \bibinfo {author} {\bibfnamefont {A.}~\bibnamefont
  {Zeilinger}},\ }\href
  {https://iopscience.iop.org/article/10.1088/1367-2630/13/5/053047} {\bibfield
   {journal} {\bibinfo  {journal} {New Journal of Physics}\ }\textbf {\bibinfo
  {volume} {13}},\ \bibinfo {pages} {053047} (\bibinfo {year}
  {2011})}\BibitemShut {NoStop}%
\bibitem [{\citenamefont {McNulty}\ \emph {et~al.}(2016)\citenamefont
  {McNulty}, \citenamefont {Pammer},\ and\ \citenamefont
  {Weigert}}]{mcnulty2016mutually}%
  \BibitemOpen
  \bibfield  {author} {\bibinfo {author} {\bibfnamefont {D.}~\bibnamefont
  {McNulty}}, \bibinfo {author} {\bibfnamefont {B.}~\bibnamefont {Pammer}}, \
  and\ \bibinfo {author} {\bibfnamefont {S.}~\bibnamefont {Weigert}},\ }\href
  {https://aip.scitation.org/doi/10.1063/1.4943301} {\bibfield  {journal}
  {\bibinfo  {journal} {Journal of Mathematical Physics}\ }\textbf {\bibinfo
  {volume} {57}},\ \bibinfo {pages} {032202} (\bibinfo {year}
  {2016})}\BibitemShut {NoStop}%
\bibitem [{\citenamefont {Wocjan}\ and\ \citenamefont
  {Beth}(2005)}]{wocjan2005new}%
  \BibitemOpen
  \bibfield  {author} {\bibinfo {author} {\bibfnamefont {P.}~\bibnamefont
  {Wocjan}}\ and\ \bibinfo {author} {\bibfnamefont {T.}~\bibnamefont {Beth}},\
  }\href@noop {} {\bibfield  {journal} {\bibinfo  {journal} {Quant. Inf.
  Comput}\ }\textbf {\bibinfo {volume} {5}},\ \bibinfo {pages} {181} (\bibinfo
  {year} {2005})}\BibitemShut {NoStop}%
\bibitem [{\citenamefont {Durt}\ \emph {et~al.}(2010)\citenamefont {Durt},
  \citenamefont {Englert}, \citenamefont {Bengtsson},\ and\ \citenamefont
  {{\.Z}yczkowski}}]{durt2010mutually}%
  \BibitemOpen
  \bibfield  {author} {\bibinfo {author} {\bibfnamefont {T.}~\bibnamefont
  {Durt}}, \bibinfo {author} {\bibfnamefont {B.-G.}\ \bibnamefont {Englert}},
  \bibinfo {author} {\bibfnamefont {I.}~\bibnamefont {Bengtsson}}, \ and\
  \bibinfo {author} {\bibfnamefont {K.}~\bibnamefont {{\.Z}yczkowski}},\ }\href
  {https://www.worldscientific.com/doi/10.1142/S0219749910006502} {\bibfield
  {journal} {\bibinfo  {journal} {International journal of quantum
  information}\ }\textbf {\bibinfo {volume} {8}},\ \bibinfo {pages} {535}
  (\bibinfo {year} {2010})}\BibitemShut {NoStop}%
\bibitem [{\citenamefont {Alber}\ and\ \citenamefont
  {Charnes}(2018)}]{alber2018mutually}%
  \BibitemOpen
  \bibfield  {author} {\bibinfo {author} {\bibfnamefont {G.}~\bibnamefont
  {Alber}}\ and\ \bibinfo {author} {\bibfnamefont {C.}~\bibnamefont
  {Charnes}},\ }\href
  {https://iopscience.iop.org/article/10.1088/1402-4896/aaecad} {\bibfield
  {journal} {\bibinfo  {journal} {Physica Scripta}\ }\textbf {\bibinfo {volume}
  {94}},\ \bibinfo {pages} {014007} (\bibinfo {year} {2018})}\BibitemShut
  {NoStop}%
\bibitem [{\citenamefont {Horodecki}\ \emph {et~al.}(2020)\citenamefont
  {Horodecki}, \citenamefont {Rudnicki},\ and\ \citenamefont
  {{\.Z}yczkowski}}]{horodecki2020five}%
  \BibitemOpen
  \bibfield  {author} {\bibinfo {author} {\bibfnamefont {P.}~\bibnamefont
  {Horodecki}}, \bibinfo {author} {\bibfnamefont {{\L}.}~\bibnamefont
  {Rudnicki}}, \ and\ \bibinfo {author} {\bibfnamefont {K.}~\bibnamefont
  {{\.Z}yczkowski}},\ }\href {https://arxiv.org/abs/2002.03233} {\bibfield
  {journal} {\bibinfo  {journal} {arXiv preprint arXiv:2002.03233}\ } (\bibinfo
  {year} {2020})}\BibitemShut {NoStop}%
\bibitem [{\citenamefont {Bengtsson}\ \emph {et~al.}(2007)\citenamefont
  {Bengtsson}, \citenamefont {Bruzda}, \citenamefont {Ericsson}, \citenamefont
  {Larsson}, \citenamefont {Tadej},\ and\ \citenamefont
  {{\.Z}yczkowski}}]{bengtsson2007mutually}%
  \BibitemOpen
  \bibfield  {author} {\bibinfo {author} {\bibfnamefont {I.}~\bibnamefont
  {Bengtsson}}, \bibinfo {author} {\bibfnamefont {W.}~\bibnamefont {Bruzda}},
  \bibinfo {author} {\bibfnamefont {{\AA}.}~\bibnamefont {Ericsson}}, \bibinfo
  {author} {\bibfnamefont {J.-{\AA}.}\ \bibnamefont {Larsson}}, \bibinfo
  {author} {\bibfnamefont {W.}~\bibnamefont {Tadej}}, \ and\ \bibinfo {author}
  {\bibfnamefont {K.}~\bibnamefont {{\.Z}yczkowski}},\ }\href
  {https://aip.scitation.org/doi/full/10.1063/1.2716990} {\bibfield  {journal}
  {\bibinfo  {journal} {Journal of mathematical physics}\ }\textbf {\bibinfo
  {volume} {48}},\ \bibinfo {pages} {052106} (\bibinfo {year}
  {2007})}\BibitemShut {NoStop}%
\bibitem [{\citenamefont {Butterley}\ and\ \citenamefont
  {Hall}(2007)}]{butterley2007numerical}%
  \BibitemOpen
  \bibfield  {author} {\bibinfo {author} {\bibfnamefont {P.}~\bibnamefont
  {Butterley}}\ and\ \bibinfo {author} {\bibfnamefont {W.}~\bibnamefont
  {Hall}},\ }\href
  {https://www.sciencedirect.com/science/article/pii/S0375960107006111}
  {\bibfield  {journal} {\bibinfo  {journal} {Physics Letters A}\ }\textbf
  {\bibinfo {volume} {369}},\ \bibinfo {pages} {5} (\bibinfo {year}
  {2007})}\BibitemShut {NoStop}%
\bibitem [{\citenamefont {Brierley}\ and\ \citenamefont
  {Weigert}(2009)}]{PhysRevA.79.052316}%
  \BibitemOpen
  \bibfield  {author} {\bibinfo {author} {\bibfnamefont {S.}~\bibnamefont
  {Brierley}}\ and\ \bibinfo {author} {\bibfnamefont {S.}~\bibnamefont
  {Weigert}},\ }\href {\doibase 10.1103/PhysRevA.79.052316} {\bibfield
  {journal} {\bibinfo  {journal} {Phys. Rev. A}\ }\textbf {\bibinfo {volume}
  {79}},\ \bibinfo {pages} {052316} (\bibinfo {year} {2009})}\BibitemShut
  {NoStop}%
\bibitem [{\citenamefont {Grassl}(2004)}]{grassl2004sic}%
  \BibitemOpen
  \bibfield  {author} {\bibinfo {author} {\bibfnamefont {M.}~\bibnamefont
  {Grassl}},\ }\href {https://arxiv.org/abs/quant-ph/0406175} {\bibfield
  {journal} {\bibinfo  {journal} {arXiv preprint quant-ph/0406175}\ } (\bibinfo
  {year} {2004})}\BibitemShut {NoStop}%
\bibitem [{\citenamefont {Jaming}\ \emph {et~al.}(2009)\citenamefont {Jaming},
  \citenamefont {Matolcsi}, \citenamefont {M{\'o}ra}, \citenamefont
  {Sz{\"o}ll{\H{o}}si},\ and\ \citenamefont {Weiner}}]{jaming2009generalized}%
  \BibitemOpen
  \bibfield  {author} {\bibinfo {author} {\bibfnamefont {P.}~\bibnamefont
  {Jaming}}, \bibinfo {author} {\bibfnamefont {M.}~\bibnamefont {Matolcsi}},
  \bibinfo {author} {\bibfnamefont {P.}~\bibnamefont {M{\'o}ra}}, \bibinfo
  {author} {\bibfnamefont {F.}~\bibnamefont {Sz{\"o}ll{\H{o}}si}}, \ and\
  \bibinfo {author} {\bibfnamefont {M.}~\bibnamefont {Weiner}},\ }\href
  {https://iopscience.iop.org/article/10.1088/1751-8113/42/24/245305}
  {\bibfield  {journal} {\bibinfo  {journal} {Journal of Physics A:
  Mathematical and Theoretical}\ }\textbf {\bibinfo {volume} {42}},\ \bibinfo
  {pages} {245305} (\bibinfo {year} {2009})}\BibitemShut {NoStop}%
\bibitem [{\citenamefont {Jaming}\ \emph {et~al.}(2010)\citenamefont {Jaming},
  \citenamefont {Matolcsi},\ and\ \citenamefont
  {M{\'o}ra}}]{jaming2010problem}%
  \BibitemOpen
  \bibfield  {author} {\bibinfo {author} {\bibfnamefont {P.}~\bibnamefont
  {Jaming}}, \bibinfo {author} {\bibfnamefont {M.}~\bibnamefont {Matolcsi}}, \
  and\ \bibinfo {author} {\bibfnamefont {P.}~\bibnamefont {M{\'o}ra}},\ }\href
  {https://link.springer.com/article/10.1007%2Fs12095-010-0023-1} {\bibfield
  {journal} {\bibinfo  {journal} {Cryptography and Communications}\ }\textbf
  {\bibinfo {volume} {2}},\ \bibinfo {pages} {211} (\bibinfo {year}
  {2010})}\BibitemShut {NoStop}%
\bibitem [{\citenamefont {Raynal}\ \emph {et~al.}(2011)\citenamefont {Raynal},
  \citenamefont {L\"u},\ and\ \citenamefont {Englert}}]{PhysRevA.83.062303}%
  \BibitemOpen
  \bibfield  {author} {\bibinfo {author} {\bibfnamefont {P.}~\bibnamefont
  {Raynal}}, \bibinfo {author} {\bibfnamefont {X.}~\bibnamefont {L\"u}}, \ and\
  \bibinfo {author} {\bibfnamefont {B.-G.}\ \bibnamefont {Englert}},\ }\href
  {\doibase 10.1103/PhysRevA.83.062303} {\bibfield  {journal} {\bibinfo
  {journal} {Phys. Rev. A}\ }\textbf {\bibinfo {volume} {83}},\ \bibinfo
  {pages} {062303} (\bibinfo {year} {2011})}\BibitemShut {NoStop}%
\bibitem [{\citenamefont {Goyeneche}(2013)}]{goyeneche2013mutually}%
  \BibitemOpen
  \bibfield  {author} {\bibinfo {author} {\bibfnamefont {D.}~\bibnamefont
  {Goyeneche}},\ }\href
  {https://iopscience.iop.org/article/10.1088/1751-8113/46/10/105301}
  {\bibfield  {journal} {\bibinfo  {journal} {Journal of Physics A:
  Mathematical and Theoretical}\ }\textbf {\bibinfo {volume} {46}},\ \bibinfo
  {pages} {105301} (\bibinfo {year} {2013})}\BibitemShut {NoStop}%
\bibitem [{\citenamefont {Batle}(2013)}]{batle2013new}%
  \BibitemOpen
  \bibfield  {author} {\bibinfo {author} {\bibfnamefont {J.}~\bibnamefont
  {Batle}},\ }\href {https://arxiv.org/abs/1312.4021} {\bibfield  {journal}
  {\bibinfo  {journal} {arXiv preprint arXiv:1312.4021}\ } (\bibinfo {year}
  {2013})}\BibitemShut {NoStop}%
\bibitem [{\citenamefont {Chen}\ and\ \citenamefont
  {Yu}(2018)}]{chen2018mutually}%
  \BibitemOpen
  \bibfield  {author} {\bibinfo {author} {\bibfnamefont {L.}~\bibnamefont
  {Chen}}\ and\ \bibinfo {author} {\bibfnamefont {L.}~\bibnamefont {Yu}},\
  }\href {https://link.springer.com/article/10.1007/s11128-018-1964-0}
  {\bibfield  {journal} {\bibinfo  {journal} {Quantum Information Processing}\
  }\textbf {\bibinfo {volume} {17}},\ \bibinfo {pages} {1} (\bibinfo {year}
  {2018})}\BibitemShut {NoStop}%
\bibitem [{\citenamefont {McNulty}\ and\ \citenamefont
  {Weigert}(2012{\natexlab{a}})}]{mcnulty2012all}%
  \BibitemOpen
  \bibfield  {author} {\bibinfo {author} {\bibfnamefont {D.}~\bibnamefont
  {McNulty}}\ and\ \bibinfo {author} {\bibfnamefont {S.}~\bibnamefont
  {Weigert}},\ }\href@noop {} {\bibfield  {journal} {\bibinfo  {journal}
  {Journal of Physics A: Mathematical and Theoretical}\ }\textbf {\bibinfo
  {volume} {45}},\ \bibinfo {pages} {135307} (\bibinfo {year}
  {2012}{\natexlab{a}})}\BibitemShut {NoStop}%
\bibitem [{\citenamefont {McNulty}\ and\ \citenamefont
  {Weigert}(2012{\natexlab{b}})}]{mcnulty2012impossibility}%
  \BibitemOpen
  \bibfield  {author} {\bibinfo {author} {\bibfnamefont {D.}~\bibnamefont
  {McNulty}}\ and\ \bibinfo {author} {\bibfnamefont {S.}~\bibnamefont
  {Weigert}},\ }\href
  {https://www.worldscientific.com/doi/abs/10.1142/S0219749912500566}
  {\bibfield  {journal} {\bibinfo  {journal} {International Journal of Quantum
  Information}\ }\textbf {\bibinfo {volume} {10}},\ \bibinfo {pages} {1250056}
  (\bibinfo {year} {2012}{\natexlab{b}})}\BibitemShut {NoStop}%
\bibitem [{\citenamefont {Bengtsson}\ and\ \citenamefont
  {Ericsson}(2005)}]{bengtsson2005mutually}%
  \BibitemOpen
  \bibfield  {author} {\bibinfo {author} {\bibfnamefont {I.}~\bibnamefont
  {Bengtsson}}\ and\ \bibinfo {author} {\bibfnamefont {{\AA}.}~\bibnamefont
  {Ericsson}},\ }\href {https://doi.org/10.1007/s11080-005-5721-3} {\bibfield
  {journal} {\bibinfo  {journal} {Open Systems \& Information Dynamics}\
  }\textbf {\bibinfo {volume} {12}},\ \bibinfo {pages} {107} (\bibinfo {year}
  {2005})}\BibitemShut {NoStop}%
\bibitem [{\citenamefont {Friedland}\ \emph {et~al.}(2013)\citenamefont
  {Friedland}, \citenamefont {Gheorghiu},\ and\ \citenamefont
  {Gour}}]{PhysRevLett.111.230401}%
  \BibitemOpen
  \bibfield  {author} {\bibinfo {author} {\bibfnamefont {S.}~\bibnamefont
  {Friedland}}, \bibinfo {author} {\bibfnamefont {V.}~\bibnamefont
  {Gheorghiu}}, \ and\ \bibinfo {author} {\bibfnamefont {G.}~\bibnamefont
  {Gour}},\ }\href {\doibase 10.1103/PhysRevLett.111.230401} {\bibfield
  {journal} {\bibinfo  {journal} {Phys. Rev. Lett.}\ }\textbf {\bibinfo
  {volume} {111}},\ \bibinfo {pages} {230401} (\bibinfo {year}
  {2013})}\BibitemShut {NoStop}%
\bibitem [{\citenamefont {Brukner}\ and\ \citenamefont
  {Zeilinger}(1999)}]{PhysRevLett.83.3354}%
  \BibitemOpen
  \bibfield  {author} {\bibinfo {author} {\bibfnamefont {{\v{C}}.}~\bibnamefont
  {Brukner}}\ and\ \bibinfo {author} {\bibfnamefont {A.}~\bibnamefont
  {Zeilinger}},\ }\href {\doibase 10.1103/PhysRevLett.83.3354} {\bibfield
  {journal} {\bibinfo  {journal} {Physical Review Letters}\ }\textbf {\bibinfo
  {volume} {83}},\ \bibinfo {pages} {3354} (\bibinfo {year}
  {1999})}\BibitemShut {NoStop}%
\bibitem [{\citenamefont {Hardy}\ \emph {et~al.}(1952)\citenamefont {Hardy},
  \citenamefont {Littlewood}, \citenamefont {P{\'o}lya}, \citenamefont
  {P{\'o}lya}, \citenamefont {Littlewood} \emph
  {et~al.}}]{hardy1952inequalities}%
  \BibitemOpen
  \bibfield  {author} {\bibinfo {author} {\bibfnamefont {G.~H.}\ \bibnamefont
  {Hardy}}, \bibinfo {author} {\bibfnamefont {J.~E.}\ \bibnamefont
  {Littlewood}}, \bibinfo {author} {\bibfnamefont {G.}~\bibnamefont
  {P{\'o}lya}}, \bibinfo {author} {\bibfnamefont {G.}~\bibnamefont
  {P{\'o}lya}}, \bibinfo {author} {\bibfnamefont {D.}~\bibnamefont
  {Littlewood}},  \emph {et~al.},\ }\href
  {https://www.cambridge.org/in/academic/subjects/mathematics/abstract-analysis/inequalities?format=PB&isbn=9780521358804}
  {\emph {\bibinfo {title} {Inequalities}}}\ (\bibinfo  {publisher} {Cambridge
  university press},\ \bibinfo {year} {1952})\BibitemShut {NoStop}%
\end{thebibliography}%
\end{document}